\documentclass{article}

\usepackage{microtype}
\usepackage{enumerate}
\usepackage{textcomp}
\usepackage{latexsym}
\usepackage{amsfonts,amssymb,amsmath,amsthm}
\usepackage{hyperref}
\usepackage{general}
\usepackage{while}
\usepackage{pars}
\usepackage{code}
\usepackage{comment}

\title{Modular Runtime Complexity Analysis of Probabilistic While Programs}
\author{
  Martin Avanzini\\
  INRIA Sophia Antipolis\\France\\
  \url{martin.avanzini@inria.fr}
  \and
  Michael Schaper\\
  University of Innsbruck\\Austria\\
  \url{michael.schaper@student.uibk.ac.at}
  \and
  Georg Moser \\
  University of Innsbruck\\Austria\\
  \url{georg.moser@uibk.ac.at}
}

\theoremstyle{plain}
\newtheorem{proposition}{Proposition}
\newtheorem{lemma}{Lemma}
\newtheorem{theorem}{Theorem}
\newtheorem{corollary}{Corollary}

\theoremstyle{definition}
\newtheorem{definition}{Definition}
\newtheorem{example}{Example}

\theoremstyle{remark}

\renewcommand{\labelitemi}{--}
\begin{document}
\maketitle

\section{Preliminaries}
\paragraph*{Multisets}
A (possibly infinite) multiset over a set $A$ is a mapping $M: A \to \N$.
The \emph{union} $\biguplus_{i\in I}M_i$
of a countably many multisets $M_i$ is defined by
\[
\pa[\Big]{\biguplus_{i\in I}M_i}(a) \defeq \sum_{i\in I}M_i(a)
\]
which forms a multiset if and only if
$\sum_{i\in I}M_i(a)$ is finite for every $a\in A$.
The \emph{sum} of a multiset $M$
with respect to $f : A \to \Realpos$ is defined by
\[
\sum_{a\in M}f(a) \defeq \sum_{a \in A} M(a) \cdot f(a)
\]
We use set-like notations for multisets:
$\varnothing$ denotes the \emph{empty} multiset $\varnothing(a) \defeq 0$,
$\mset{a_i \mid i \in I}$
is the multiset $M$ with $M(a) = |\set{i\in I \mid a_i = a}|$,
and
$\mset{a_1, \dots, a_n}$
is its special case where $I = \set{1,\dots,n}$ is finite.

\paragraph*{Weighted Abstract Reduction Systems}
A \emph{weighted ARS} over state space $A$ is a ternary relation
${\to} \subseteq A \times A \times \Realnonneg$.
We write $a \to\AT{c} b$ meaning $\tp{a,b,c} \in {\to}$.
We define the ARS ${\to^c} \subseteq A \times A$ \emph{induced by $\to$ at cost $c \in \Realnonneg$},
by the following inference rules:
\[
\infer{a \to^0 a}{}\qquad
\infer{a \to^c b}{a \to\AT{c} b}\qquad
\infer{a \to^{c+d} b}{a \to^c a' \quad a'\to^d b}
\]
We say $a \in A$ is a \emph{normal form} with respect to $\to$ if
there exists no $b \in A$ with $a \to^+ b$.
The set of normal forms with respect to $\to$ is denoted by $\NF(\to)$.
The \emph{potential} of $a \in A$ with respect to $\to$ is defined by
$\pot[\to]{a} \defeq \sup\set{c \mid \exists b.\ a \to^c b}$.
The weighted ARS $\to$ is called \emph{strongly bounded} on $S \subseteq A$, $\SB_\to(S)$,
if for every $a \in S$, there exists $p \in \Realnonneg$ such that
$a \to^c b$ implies $c \le p$;
This is equivalent to saying that $\pot[\to]{a} \leq \infty$ for every $a \in S$.

\paragraph*{Weighted Probabilistic Abstract Reduction Systems}

A \emph{multidistribution} on a set $A$ is a multiset $\mdist$
of pairs of $a\in A$ and
$0 < p \le 1$, written $p:a$, satisfying
$\sz{\mdist} \defeq \sum_{p:a \in \mdist} p \ \le\ 1$.
We denote the set of multidistributions on $A$ by $\MDist{A}$.
Multidistributions are closed under \emph{convex multiset unions}
$\biguplus_{i\in I} p_i\cdot\mdist_i \defeq \sum_{i\in I}p_i \cdot \sz{\mdist_i} \leq 1$
for every finite or countable infinite index set $I$ and
probabilities $p_i > 0$ with $\sum_{i\in I}p_i \le 1$,
where \emph{scalar multiplication} is defined by
$p \cdot \mset{q_i:a_i \mid i \in I} \defeq \mset{p\cdot q_i : a_i \mid i \in I}$
for $0 < p \leq 1$.

The restriction of a multidistribution $\mdist \in \MDist{A}$ to a set $P \subseteq A$
is defined by $\restrict{\mdist}{P} \defsym \prms{p : a \mid p : a \in \mdist, a \in P}$.
For a function $f : A \to B$, we denote by $\fmap{f}$
its homomorphic extension $\fmap{f}: \MDist{A} \to \MDist{B}$ defined by
\[
  \fmap{f}\bigl(\prms{p_i:a_i \mid i \in I}\bigr)
  \defsym \prms{p_i:f(a_i) \mid i \in I}\tpkt
\]

For $\mdist \in \MDist{A}$, we define
the \emph{expectation} of a function $f \ofdom A \to \Realposinf$
as $\E{\mdist}(f) \defsym \sum_{p:a \in \mdist} p \cdot f(a)$.
Notice that
$\E{\biguplus_{i \in I} p_i \cdot \mdist_i}(f) = \sum_{i \in I} p_i \cdot \E{\mdist_i}(f)$.

\begin{definition}
  A \emph{weighted probabilistic} ARS over $A$ is a set ${\to} \subseteq A \times \MDist{A} \times \Realnonneg$.
  As before, we may write $a \to\AT{c} \mu$ for $\tp{a,\mu,c} \in {\to}$.
  We define the weighted ARS ${\wars\to}$ over $\MDist{A}$ \emph{induced by} $\to$ as follows:
  \[
    \infer{\mu \wars\to\AT{0} \mu}{}
    \qquad
    \infer{\mset{1:a} \wars\to\AT{c} \mu}{a \to\AT{c} \mu}
    \qquad
    \infer{
      \biguplus_{i\in I}p_i\cdot \mu_i \wars\to\AT{c} \biguplus_{i\in I}p_i\cdot \nu_i
    }{
      \forall i\in I.\ \mu_i \wars\to\AT{c_i} \nu_i
      & c = \sum_{i\in I} p_i\cdot c_i
    }
  \]
\end{definition}

For a weighted probabilistic ARS $\to$, let us define $\pot[\to]{a} \defeq \pot[\wars\to]{\mset{1:a}}$.
A weighted probabilistic ARS $\to$ over $A$ is \emph{strongly bounded} on a set $S \subseteq A$
if $\SB_{\wars\to}(\set{\mset{1:a} \mid a \in S})$, i.e.,
$\pot[\to]{a} < \infty$ for all $S \in A$.

\section{Probabilistic While}

We consider an imperative language $\lang$ in the spirit of Dijkstra's \emph{Guarded Command Language} \cite{Dijkstra75},
endowed with primitives for sampling from discrete distributions as well as non-deterministic and probabilistic choice.
Let $\Var$ denote a finite set of integer-valued variables $\var,\vartwo,\dots$.
We denote by $\Store \defsym \Var \to \Z$ the set of \emph{stores},
that associate variables with their integer contents.
The syntax of \emph{program commands} $\Cmd$ over $\Var$ is given by the following grammar.
\newcommand{\lbl}[1]{(#1)}
\newcommand{\desc}[1]{\text{\emph{#1}}}
\begin{align*}
  \lbl{\Cmd} &  & \cmd,\cmdtwo & \bnfdef \skipc                             &  & \desc{effectless operation}     \\
             &  &              & \mid  \tickc{r}                            &  & \desc{resource consumption}     \\
             &  &              & \mid  \haltc                               &  & \desc{termination}              \\
             &  &              & \mid  \var \passign \expd                  &  & \desc{probabilistic assignment} \\
             &  &              & \mid  \ifc[\bexptwo]{\bexp}{\cmd}{\cmdtwo} &  & \desc{conditional}              \\
             &  &              & \mid  \whilec[\bexptwo]{\bexp}{\cmd}       &  & \desc{while loop}               \\
             &  &              & \mid  \ndc{\cmd}{\cmdtwo}                  &  & \desc{non-deterministic choice} \\
             &  &              & \mid  \pc{p}{\cmd}{\cmdtwo}                &  & \desc{probabilistic choice}     \\
             &  &              & \mid  \cmd; \cmdtwo                        &  & \desc{sequential composition}.
\end{align*}
In this grammar,
$\bexp \in \BExp$ denotes a \emph{Boolean expression} over $\Var$ and
$\expd \in \DExp$ an \emph{Integer-valued distribution expression} over $\Var$.
With $\evdist{\cdot} \ofdom \DExp \to \Store \to \Dist{\Z}$
we denote the evaluation functions of distribution expressions, i.e.,
$\evdist{\expd}(\st)$ gives the result of evaluating $\expd$
under the current store $\st$.
For Boolean expressions $\bexp \in \BExp$ and $\st \in \Store$, we
indicate with $\st \entails \bexp$ that $\bexp$ holds when
the variables in $\bexp$ take values according to $\st$.

Program commands are fairly standard. The command $\skipc$ is a no-op,
$\haltc$ terminates the execution.
The command $\tickc{r}$ consumes $r \in \Q^+$ resource units.
The command $\var \passign \expd$ assigns a value sampled from $\expd(\st)$ to $\var$,
for $\st$ the current store.
The usual non-probabilistic assignment $\var \passign \exp$
for expressions $\exp \in \Exp$ is recovered by
the probabilistic assignment $\var \passign \expd_\exp$,
where $\expd_\exp(\st) \defsym \{1:\evexp{\exp}(\st)\}$.
The commands $\ifc[\bexptwo]{\bexp}{\cmd}{\cmdtwo}$ and $\whilec[\bexptwo]{\bexp}{\cmd}$ have
the usual semantics, with $\bexptwo$ an assertion that has to hold when entering the command.
We abbreviate $\ifc[\bexptwo]{\bexp}{\cmd}{\cmdtwo}$ and $\whilec[\bexptwo]{\bexp}{\cmd}$
by $\ifc{\bexp}{\cmd}{\cmdtwo}$ and $\whilec{\bexp}{\cmd}$ when $\bexp$ is the trivial assertion $\top$
that is always true.
The command $\ndc{\cmd}{\cmdtwo}$ executes either $\cmd$ or $\cmdtwo$, in a non-deterministic fashion.
Our analysis takes a demonic view on non-determinsmn, assuming that the branch with worst-case resource consumption is taken.
In contrast, the probabilistic choice $\pc{p}{\cmd}{\cmdtwo}$ executes $\cmd$ with probability $p \in \Prob$ and with probability
$1-p$ the command $\cmdtwo$.

\begin{example}[Random Walk]
  The program
  \begin{lstlisting}[style=pwhile,emph={x}]
    while (x > 0) { tick(1); {x := x + 1} $\pci{p}$ {x := x - 1}  }
  \end{lstlisting}
  describes a random walk over $\N$. The resource metric taken, via the command \pwhile!tick(1)! in the loop body, gives the number of loop iterations.
\end{example}

\subsection{Small Step Operational Semantics}
\begin{figure}
  \centering
  \begin{framed}
    \small
    \[
      \Infer[step][Skip]{\cfg{\skipc}{\st} \whilePARS[0] \st}{}
      \quad
      \Infer[step][Tick]{\cfg{\tickc{r}}{\st} \whilePARS[r] \st}{}
      \quad
      \Infer[step][Halt]{\cfg{\haltc}{\st} \whilePARS[0] \aborted}{}
    \]
    \[
      \Infer[step][Assign]{
        \cfg{\var \passign \expd}{\st}
        \whilePARS[0] \prms{ d(\st)(i) : \upd{\st}{\var}{i} \mid i \in \Z, d(\st)(i) > 0}
      }{}
    \]
    \[
      \Infer[step][IfTrue]
      {\cfg{\ifc[\bexptwo]{\bexp}{\cmd}{\cmdtwo}}{\st} \whilePARS[0] \cfg{\cmd}{\st}}
      { \st \entails \bexptwo \land \bexp }
    \]
    \[
      \Infer[step][IfFalse]
      {\cfg{\ifc[\bexptwo]{\bexp}{\cmd}{\cmdtwo}}{\st} \whilePARS[0] \cfg{\cmdtwo}{\st}}
      { \st \entails \bexptwo \land \neg\bexp }
    \]
    \[
      \Infer[step][IfFail]
      {\cfg{\ifc[\bexptwo]{\bexp}{\cmd}{\cmdtwo}}{\st} \whilePARS[0] \aborted}
      { \st \entails \neg\bexptwo }
    \]
    \[
      \Infer[step][WhileStep]
      {\cfg{\whilec[\bexptwo]{\bexp}{\cmd}}{\st} \whilePARS[0] \cfg{\cmd; \whilec[\bexptwo]{\bexp}{\cmd}}{\st}}
      { \st \entails \bexptwo \land \bexp }
    \]
    \[
      \Infer[step][WhileFin]
      {\cfg{\whilec[\bexptwo]{\bexp}{\cmd}}{\st} \whilePARS[0] \st}
      { \st \entails \bexptwo \land \neg\bexp }
    \]
    \[
      \Infer[step][WhileFail]
      {\cfg{\whilec[\bexptwo]{\bexp}{\cmd}}{\st} \whilePARS[0] \aborted}
      { \st \entails \neg\bexptwo }
    \]
    \[
      \Infer[step][ChoiceL]
      {\cfg{\ndc{\cmd}{\cmdtwo}}{\st} \whilePARS[0] \cfg{\cmd}{\st}}{}
      \quad
      \Infer[step][ChoiceR]
      {\cfg{\ndc{\cmd}{\cmdtwo}}{\st} \whilePARS[0] \cfg{\cmdtwo}{\st}}{}
    \]
    \[
      \Infer[step][ProbChoice]
      {\cfg{\pc{p}{\cmd}{\cmdtwo}}{\st} \whilePARS[0] \prms{ p : \cfg{\cmd}{\st}, 1-p : \cfg{\cmdtwo}{\st}}}{}
    \]
    \[
    \begin{array}{c}
      \Infer[step][Compose]
      {\cfg{\cmd;\cmdtwo}{\st} \whilePARS[r] \fmap{\evstep{\cmdtwo}}(\mdist)}
      {\cfg{\cmd}{\st} \whilePARS[r] \mdist}
    \end{array}
    \text{ where }
    \evstep{\cmdtwo}(\conf) \defsym
    \begin{cases}
      \cfg{\cmd;\cmdtwo}{\st} & \text{if $\conf = \cfg{\cmd}{\st}$} \\
      \cfg{\cmdtwo}{\st} & \text{if $\conf = \st \in \Store$} \\
      \aborted & \text{if $\conf = \aborted$.}
    \end{cases}
    \]
  \end{framed}
  \caption{One-step reduction relation as a weighted probabilistic ARS.}
  \label{fig:pwhile}
\end{figure}

We give small step operational semantics for our language
via a weighted probabilistic ARS ${\whilePARS}$ over \emph{configurations}
\[
  \Conf \defsym (\Cmd \times \Store) \cup \Store \cup \{\aborted\}
  \tpkt
\]
Elements $(\cmd,\st) \in \Conf$ are called \emph{active} and denoted by $\cfg{\cmd}{\st}$.
Such an active configuration signals that the command $\cmd$ is to be executed under
the current store $\st$, whereas $\st \in \Conf$ and $\aborted \in \Conf$ indicate that the computation
has halted. The former case gives the final store, whereas the later signals that the command terminated abnormally.
The probabilistic ARS $\whilePARS$ is depicted in Figure~\ref{fig:pwhile}.
The rules of this system reflect the
operational semantics that were informally outlined above.
To avoid syntactic overhead, we identify configurations $\conf$ with dirac multidistribution
$\mset{1 : \conf}$. Thereby, a rule $\conf_1 \whilePARS[w] \mset{1 : \conf_2}$ without probabilistic effect
can be simply denoted by $\conf_1 \whilePARS[w] \conf_2$.

\begin{definition}
  Let $\whileSTEP[]$ be the ARS over $\MDist{\Conf}$ associated with the probabilistic ARS $\whilePARS$ from Figure~\ref{fig:pwhile}.
  The expected cost function $\ec{\cdot} : \Cmd \to \Store \to \Realposinf$
  is defined by
  \[
    \ec{\cmd}(\st) \defsym \sup\{ w \mid \cfg{\cmd}{\st} \whileSEQ[w] \mdist \}
  \]
  The expected value function
  $\ev{\cdot} : \Cmd \to \Store \to (\Store \to \Realposinf) \to \Realposinf$ is given by
  \[
    \ev{\cmd}(\st)(f) \defsym \sup\{ f(\restrict{\mdist}{\Store})   \mid \cfg{\cmd}{\st} \whileSEQ[w] \mdist \}
  \]
\end{definition}

\section{Expectation Transformers}
\cite{KaminskiKMO16} introduce the transformer $\mathsf{ert}[\cmd]$ for
reasoning about the expected runtime of probabilistic while programs.
In this section, we suite this transformer to two transformers
$\ect{\cmd}$ and $\evt{\cmd}$ that compute the expected cost and expected value function
of the program $\cmd$, respectively.
We then prove them sound
with respect to the small step operational semantics introduced in the last section.

Let $\Expect \defsym \Store \to \Realposinf$ be the set of \emph{expectation functions}.
We extend functions $f \ofdom (\Realposinf)^k \to \Realposinf$ pointwise on expectations
and denote these in bold face,
e.g., for each $r \in \Realposinf$ we have a constant function $\mathbf{r}(\st) \defsym r$,
$f \plusf g \defsym \lambda \st. f(\st) + g(\st)$ for $f,g \in \Expect$ etc.
For $\bexp \in \BExp$ we use Iverson's bracket $\brac{\bexp}$ to denote the expectation function
$\brac{\bexp}(\st) \defsym 1$ if $\st \entails \bexp$, and $\brac{\bexp}(\st) \defsym 0$ otherwise.
In particular, $\brac{\btrue}(\st)$ and $\brac{\bfalse}(\st)$ are the constant functions
that evaluate to $1$ and $0$, respectively.
Let $\leqf$ be the point-wise ordering on $\Expect$, i.e.,
\[
  f \leqf g \defiff f(\st) \leq g(\st)
\]
for all $\st \in \Store$. We denote by $\geqf$ the inverse of $\leqf$.
The proof of the following is standard.
\begin{proposition}
  $(\Expect, {\leqf})$ is an $\omega$-CPO, i.e., it is a poset
  in which every $\omega$-chain $f_0 \leqf f_1 \leqf f_2 \leqf \cdots$
  has a supremum in $\Expect$.
  The bottom and top element are $\zerof$ and $\inftyf$, respectively.
  The supremum of an $\omega$-chain $(f_n)_{n \in \N}$ is given point-wise: $\sup_{n \in \N} f_n \defsym \lambda \st. \sup_{n \in \N} f_n(\st)$.
\end{proposition}

\begin{figure}
  \centering
  \begin{framed}
    \vspace{-\baselineskip}
    \begin{align*}
      \et{\skipc}(f)                               & \defsym f \\
      \et[c]{\tickc{r}}(f)                         & \defsym \bracf{c} \mulf \constf{r} \plusf f \\
      \et{\haltc}(f)                               & \defsym \zerof \\
      \et{\var \passign \expd}(f)                  & \defsym \lambda \st. \E{\expd(\st)}(\lambda i. f(\upd{\st}{\var}{i})) \\
      \et{\ifc[\bexptwo]{\bexp}{\cmd}{\cmdtwo}}(f) & \defsym \bracf{\bexptwo \land \bexp} \mulf \et{\cmd}(f) \plusf \bracf{\bexptwo \land \neg\bexp} \mulf \et{\cmdtwo}(f) \\
      \et{\whilec[\bexptwo]{\bexp}{\cmd}}(f)       & \defsym \lfp F. \bracf{\bexptwo \land \bexp} \mulf \et{\cmd}(F) \plusf \bracf{\bexptwo \land \neg\bexp} \mulf f \\
      \et{\ndc{\cmd}{\cmdtwo}}(f)                  & \defsym \maxf(\et{\cmd}(f),\et{\cmdtwo}(f)) \\
      \et{\pc{p}{\cmd}{\cmdtwo}}(f)                & \defsym \constf{p} \mulf \et{\cmd}(f) \plusf (\constf{1-p}) \mulf \et{\cmdtwo}(f) \\
      \et{\cmd; \cmdtwo}(f)                        & \defsym \et{\cmd}(\et{\cmdtwo}(f))
    \end{align*}
  \end{framed}
  \caption{Definition of expectation transformer $\et{\cdot}$.}
  \label{fig:et}
\end{figure}

We arrive at the definition of the \emph{expectation transformers}
$\ect{\cdot}$ and $\evt{\cdot}$, both of type
$\Cmd \to \Expect \to \Expect$.
The definition of $\ect{\cmd}$ and $\evt{\cmd}$
coincide up to the case where $\cmd = \tickc{r}$,
the former taking into account the cost $r$ while the latter is ignoring it.
The definition of $\ect{\cdot}$ and $\evt{\cdot}$ is
given in terms of a transformer $\et[c]{\cdot} \ofdom \Cmd \to \Expect \to \Expect$,
parameterised in a Boolean $c$ that governs the treatment of $\tickc{r}$.
We set
\begin{align*}
  \ect{\cdot} & \ofdom \Cmd \to \Expect \to \Expect
  & \evt{\cdot} & \ofdom \Cmd \to \Expect \to \Expect \\
  \ect{\cmd} & \defsym \et[\btrue]{\cmd}
  & \evt{\cmd} & \defsym \et[\bfalse]{\cmd}
\end{align*}
Informally, the transformer $\et[c]{\cdot}$ is defined in continuation style.
If $f \ofdom \Expect$ gives the cost of executing a program fragment $\cmdtwo$,
then $\ect{\cmd}(f)$ gives the cost of first $\cmd$ and then $\cmdtwo$.
Consequently, the cost of running $\cmd$ is given by $\ect{\cmd}(\zerof)$.
Likewise for $\evt{\cmd}$, where however $f$ denotes the function
applied to final states.
Finally, we note that $\evt{\cmd}$ coincides with the
weakest precondition transformer $\mathsf{wp}[\cmd]$ of \cite{OlmedoKKM16}
on \emph{fully probabilistic programs}, i.e., those without non-deterministic choice.
In contrast to $\evt{\cmd}$, $\mathsf{wp}[\cmd]$ minimises over non-deterministic choice,
i.e., $\mathsf{wp}[\cmd](f)$ gives the minimal expected value of $f$ over
all non-deterministic choices performed by $\cmd$. This is sensible
if $f$ is a predicate, i.e., $f \ofdom \Store \to \{0,1\}$ and thus
$\mathsf{wp}[\cmd](f)$ gives the least probability that $f$ holds
in a terminal state. In contrast, we are interested in the maximal value
that $f$ can take along all non-deterministic choices.

In Figure~\ref{fig:et}, $\lfp F.e$ denotes the least fixed point of the function
$\lambda F.e \ofdom \Expect \to \Expect$.
As the transfomer $\et[c]{\cmd}$ is $\omega$-continuous, $\et[c]{\cmd}$ is well-defined:
\begin{lemma}\label{l:et-continuous}
  For every $\omega$-chain $f_0 \leqf f_1 \leqf f_2 \leqf \cdots$ of expectations,
  \[
    \et[c]{\cmd}(\sup_{n \in \N} f_n) = \sup_{n \in \N} \et[c]{\cmd}(f_n)
    \tpkt
  \]
\end{lemma}
\begin{proof}
  By a standard induction on the structure of $\cmd$.
\end{proof}

As every continuous function is also monotone, we get:
\begin{lemma}\label{l:et-monotone}
  \[
    f \leqf g \IImp \et[c]{\cmd}(f) \leqf \et[c]{\cmd}(g)
    \tpkt
  \]
\end{lemma}

Towards our soundness results, we show that $\et[c]{\cmd}$ decreases along reductions
starting from $\cfg{\cmd}{\st}$,
taking into account the cost of steps when $c = \btrue$.
To formalise this, we first extend
the expectation transformer $\et{\cdot} \ofdom \Expect \to \Expect$ to a function $\e[c] \ofdom \Expect \to (\Conf \to \Realposinf)$ as follows:
\begin{align*}
  \e[c](f)(\cfg{\cmd}{\st}) & \defsym \et[c]{\cmd}(f)(\st)
  & \e[c](f)(\st) & \defsym f(\st)
  & \e[c](f)(\aborted) & \defsym 0 \tpkt
\end{align*}

\begin{lemma}\label{l:e-decrease-pars}
  \[
    \cfg{\cmd}{\st} \whilePARS[w] \mdist \IImp \e[c](f)(\cfg{\cmd}{\st}) \geq \brac{c} \mulf w + \E{\mdist}(\e[c](f))
    \tpkt
  \]
\end{lemma}
\begin{proof}
  \newcommand{\ei}{\e[c](f)}
  The proof is by induction on the definition of the probabilistic ARS $\whilePARS$.
  The case of a non-probabilistic transition $\cfg{\cmd}{\st} \whilePARS[w] \conf$
  amounts to showing $\ei(\cfg{\cmd}{\st}) \geq \brac{c} \mulf w + \ei(\conf)$.
  \begin{proofcases}
    \case{$\cfg{\skipc}{\st} \whilePARS[0] \st$}
    Then $\ei(\cfg{\skipc}{\st}) = f(\st) = \brac{c} \cdot 0 + \ei(\st)$.

    \case{$\cfg{\tickc{r}}{\st} \whilePARS[r] \st$}
    Then $\ei(\cfg{\tickc{r}}{\st}) = f(\st) = \brac{c} \cdot r + \ei(\st)$.

    \case{$\cfg{\haltc}{\st} \whilePARS[0] \aborted$}
    Then $\ei(\cfg{\haltc}{\st}) = 0 = \brac{c} \cdot 0 + \ei(\aborted)$.

    \case{$\cfg{\var \passign \expd}{\st} \whilePARS[0] \mdist$ where $\mdist = \prms{ d(\st)(i) : \upd{\st}{\var}{i} \mid i \in \Z, d(\st)(i) > 0}$}
    Then
    \begin{align*}
      \ei(\cfg{\var \passign \expd}{\st})
      & = \E{\expd(\st)}(\lambda i. f(\upd{\st}{\var}{i})) \\
      & = \E{\mdist}(f) \\
      & \mathpar{$\ei(\st') = f(\st')$ for all $\st' \in \Store$} \\
      & = \brac{c} \cdot 0 + \E{\mdist}(\ei)
      \tpkt
    \end{align*}

    \case{$\cfg{\ifc[\bexptwo]{\bexp}{\cmd}{\cmdtwo}}{\st} \whilePARS[0] \cfg{\cmd}{\st}$ where $\st \entails \bexptwo \land \bexp$}
    Then
    \begin{align*}
      \hangeq{\ei(\cfg{\ifc[\bexptwo]{\bexp}{\cmd}{\cmdtwo}}{\st})}
      & = \brac{\bexptwo \land \bexp}(\st) \cdot \et{\cmd}(f)(\st) + \brac{\bexptwo \land \neg\bexp}(\st) \cdot \et{\cmdtwo}(f)(\st) \\
      & \mathpar{$\st \entails \bexptwo \land \bexp$, hence $\brac{\bexptwo \land \bexp}(\st) = 1$ and $\brac{\bexptwo \land \neg\bexp}(\st) = 0$} \\
      & = \et{\cmd}(f)(\st) = \brac{c} \cdot 0 + \et{\cmd}(f)(\st) \tpkt
    \end{align*}

    \case{$\cfg{\ifc[\bexptwo]{\bexp}{\cmd}{\cmdtwo}}{\st} \whilePARS[0] \cfg{\cmdtwo}{\st}$ where $\st \entails \bexptwo \land \neg\bexp$}
    This case follows as above, using that $\brac{\bexptwo \land \bexp}(\st) = 0$ and $\brac{\bexptwo \land \neg\bexp}(\st) = 1$
    holds since $\st \entails \bexptwo \land \neg\bexp$.

    \case{$\cfg{\ifc[\bexptwo]{\bexp}{\cmd}{\cmdtwo}}{\st} \whilePARS[0] \aborted$ where $\st \entails \neg\bexptwo$}
    The assumption yields $\brac{\bexptwo \land \bexp}(\st) = 0$ and $\brac{\bexptwo \land \neg\bexp}(\st) = 0$, and hence
    \[
      \ei(\cfg{\ifc[\bexptwo]{\bexp}{\cmd}{\cmdtwo}}{\st}) = 0 = \brac{c} \cdot 0 + \ei(\aborted) \tpkt
    \]
    \case{$\cfg{\whilec[\bexptwo]{\bexp}{\cmd}}{\st} \whilePARS[0] \cfg{\cmd; \whilec[\bexptwo]{\bexp}{\cmd}}{\st}$ where $\st \entails \bexptwo \land \bexp$}
    Define
    \[
      F_f(g)(\st') \defsym \brac{\bexptwo \land \bexp}(\st') \cdot \et[c]{\cmd}(g)(\st') + \brac{\bexptwo \land \neg\bexp}(\st') \cdot f(\st')
      \tkom
    \]
    thus
    $\et[c]{\whilec[\bexptwo]{\bexp}{\cmd}}(f) = \lfp F_f = F_f (\lfp F_f)$. Thus
    \begin{align*}
      \hangeq{\ei(\cfg{\whilec[\bexptwo]{\bexp}{\cmd}}{\st})}
        & = F_f (\lfp F_f) (\st) \\
        & = \brac{\bexptwo \land \bexp}(\st) \cdot \et[c]{\cmd}(\lfp F_f)(\st) + \brac{\bexptwo \land \neg\bexp}(\st) \cdot f(\st) \\
        & \mathpar{$\st \entails \bexptwo \land \bexp$, hence $\brac{\bexptwo \land \bexp}(\st) = 1$ and $\brac{\bexptwo \land \neg\bexp}(\st) = 0$ } \\
        & = \et[c]{\cmd}(\lfp F_f)(\st) \\
        & = \et[c]{\cmd}(\et[c]{\whilec[\bexptwo]{\bexp}{\cmd}}(f))(\st) \\
        & = \et[c]{\cmd;\whilec[\bexptwo]{\bexp}{\cmd}}(\st) \\
        & = \brac{c} \cdot 0 + \ei(\cfg{\cmd;\whilec[\bexptwo]{\bexp}{\cmd}}{\st})
        \tpkt
    \end{align*}
    \case{$\cfg{\whilec[\bexptwo]{\bexp}{\cmd}}{\st} \whilePARS[0] \st$ where $\st \entails \bexptwo \land \neg\bexp$}
    Reasoning as above, using
    $\brac{\bexptwo \land \bexp}(\st) = 0$ and $\brac{\bexptwo \land \neg\bexp}(\st) = 1$
    we have
    \[
      \ei(\cfg{\whilec[\bexptwo]{\bexp}{\cmd}}{\st}) = f(\st) = \brac{c} \cdot 0 + \ei(\st)
      \tpkt
    \]
    \case{$\cfg{\whilec[\bexptwo]{\bexp}{\cmd}}{\st} \whilePARS[0] \aborted$ where $\st \entails \neg\bexptwo$}
    As above, using
    Reasoning as above, using
    $\brac{\bexptwo \land \bexp}(\st) = 0$ and $\brac{\bexptwo \land \neg\bexp}(\st) = 0$
    we get
    \[
      \ei(\cfg{\whilec[\bexptwo]{\bexp}{\cmd}}{\st}) = 0 = \brac{c} \cdot 0 + \ei(\aborted)
      \tpkt
    \]

    \case{$\cfg{\ndc{\cmd}{\cmdtwo}}{\st} \whilePARS[0] \cfg{\cmd}{\st}$}
    Then
    \begin{align*}
      \ei(\cfg{\ndc{\cmd}{\cmdtwo}}{\st})
      & = \max(\et[c]{\cmd}(f)(\st),\et[c]{\cmdtwo}(f)(\st)) \\
      & \geq \et[c]{\cmd}(f)(\st) = \brac{c} \cdot 0 + \ei(\cfg{\cmd}{\st})
      \tpkt
    \end{align*}
    \case{$\cfg{\ndc{\cmd}{\cmdtwo}}{\st} \whilePARS[0] \cfg{\cmdtwo}{\st}$}
    This case follows as the previous one.
    \case{$\cfg{\pc{p}{\cmd}{\cmdtwo}}{\st} \whilePARS[0] \mdist$ where $\mdist = \prms{ p : \cfg{\cmd}{\st}, 1-p : \cfg{\cmdtwo}{\st}}$}
    Then
    \begin{align*}
      \ei(\cfg{\pc{p}{\cmd}{\cmdtwo}}{\st})
      & = p \cdot \et[c]{\cmd}(f)(\st) + (1-p) \cdot \et[c]{\cmdtwo}(f)(\st) \\
      & = p \cdot \ei(\cfg{\cmd}{\st}) + (1-p) \cdot \ei(\cfg{\cmdtwo}{\st}) \\
      & = \E{\mdist}(\ei)
        \tpkt
    \end{align*}
    \case{$\cfg{\cmd;\cmdtwo}{\st} \whilePARS[r] \fmap{\evstep{\cmdtwo}}(\mdist)$ where $\cfg{\cmd}{\st} \whilePARS[r] \mdist$}
    We first show
    \begin{equation}
      \label{l:e-decrease-pars:1}
      \e[c](\et[c]{\cmdtwo}(f)) = \e[c](f) \compose \evstep{\cmdtwo}
      \tkom
    \end{equation}
    by case analysis:
    \begin{align*}
      \e[c](\et[c]{\cmdtwo}(f))(\cfg{\cmd'}{\st'})
      & = \et[c]{\cmd'}(\et[c]{\cmdtwo}(f))(\st') \\
      & = \et[c]{\cmd';\cmdtwo}(f)(\st') \\
      & = \e[c](f)(\cfg{\cmd';\cmdtwo}{\st'})
        = \e[c](f)(\evstep{\cmdtwo}(\cfg{\cmd'}{\st'}))
        \tspkt \\
      \e[c](\et[c]{\cmdtwo}(f))(\st')
      & = \et[c]{\cmdtwo}(f)(\st') \\
      & = \e[c](f)(\cfg{\cmdtwo}{\st'})
        = \e[c](f)(\evstep{\cmdtwo}(\st'))
        \tspkt \\
      \e[c](\et[c]{\cmdtwo}(f))(\aborted)
      & = 0
        = \e[c](f)(\aborted)
        = \e[c](f)(\evstep{\cmdtwo}(\aborted))
        \tpkt
    \end{align*}
    Consequently,
    \begin{align*}
      \ei(\cfg{\cmd;\cmdtwo}{\st})
      & = \et[c]{\cmd}(\et[c]{\cmdtwo}(f))(\st) \\
      & = \e[c](\et[c]{\cmdtwo}(f))(\cfg{\cmd}{\st}) \\
      & \mathpar{induction hypothesis}\\
      & \geq \brac{c} \cdot r + \E{\mdist}(\e[c](\et[c]{\cmdtwo}(f))) \\
      & \mathpar{Equation \eqref{l:e-decrease-pars:1}} \\
      & = \brac{c} \cdot r + \E{\mdist}(\e[c](f) \compose \evstep{\cmdtwo}) \\
      & = \brac{c} \cdot r + \E{\fmap{\evstep{\cmdtwo}}(\mdist)}(\e[c](f))
        \tpkt
        \qedhere
    \end{align*}
  \end{proofcases}
\end{proof}

\begin{lemma}\label{l:e-decrease-step}
  \[
    \mdist \whileSTEP[w] \mdisttwo
    \IImp \E{\mdist}(\e[c](f)) \geq \brac{c} \cdot w + \E{\mdisttwo}(\e[c](f))
    \tpkt
  \]
\end{lemma}
\begin{proof}
  The proof is by induction on the definition of $\whileSTEP[w]$:
  \begin{proofcases}
    \case{$\mdist \whileSTEP[0] \mdist$}
    This case trivially holds.

    \case{$\mdist = \mset{1:\conf} \whileSTEP[w] \mdisttwo$ where $\conf \whileSTEP[w] \mdisttwo$}
    Then $\conf = \cfg{\cmd}{\st}$ for some command $\cmd$ and store $\st$.
    We conclude this case with Lemma~\ref{l:e-decrease-pars}.

    \case{$\mdist = \biguplus_{i\in I} p_i \cdot \mdist_i \whileSTEP[\sum_{i \in I} p_i \cdot w_i] \biguplus_{i\in I} p_i \cdot \mdisttwo_i = \mdisttwo$
      where $\mdist_i \pTO[w_i] \mdisttwo_i$ for all $i \in I$}
    Then
    \begin{align*}
      \E{\mdist}(\e[c](f))
      & = \sum_{i \in I} p_i \cdot \E{\mdist_i}(\e[c](f)) \\
      & \mathpar{induction hypothesis} \\
      & \geq \sum_{i \in I} p_i \cdot (\brac{c} \cdot w_i + \E{\mdisttwo_i}(\e[c](f))) \\
      & = \sum_{i \in I} p_i \cdot \brac{c} \cdot w_i + \sum_{i \in I} p_i \cdot \E{\mdisttwo_i}(\e[c](f)) \\
      & \mathpar{$\mdisttwo = \biguplus_{i \in I} p_i \cdot \mdisttwo_i$}\\
      & = \brac{c} \cdot \sum_{i \in I} p_i \cdot w_i + \E{\mdisttwo}(\e[c](f))
        \tpkt \qedhere
    \end{align*}
  \end{proofcases}
\end{proof}

\begin{theorem}\label{t:e-decrease-seq}
  \[
    \mdist \whileSEQ[w] \mdisttwo
    \IImp \E{\mdist}(\e[c](f)) \geq \brac{c} \cdot w + \E{\mdisttwo}(\e[c](f)) \tpkt
  \]
\end{theorem}
\begin{proof}
  The proof is by induction on the definition of $\whileSTEP[w]$:
  \begin{proofcases}
    \case{$\mdist \whileSEQ[0] \mdist$}
    This case trivially holds.

    \case{$\mdist \whileSEQ[w] \mdisttwo$ where $\mdist \whileSTEP[w] \mdisttwo$}
    This case follows from Lemma~\ref{l:e-decrease-step}.

    \case{$\mdist \whileSEQ[w_1 + w_2] \mdisttwo$ where $\mdist \whileSEQ[w_1] \mdist'$ and $\mdist' \whileSEQ[w_2] \mdisttwo$}
    This case is a direct consequence of the induction hypothesis. \qedhere
  \end{proofcases}
\end{proof}
\begin{corollary}[Soundness of Expectation Transformers]\label{c:et:sound}
  \begin{enumerateenv}
  \item $\ec{\cmd}(\st) \leq \ect{\cmd}(\zerof)(\st)$.
  \item $\ec{\cmd}(\st)(f) \leq \evt{\cmd}(f)(\st)$.
  \end{enumerateenv}
\end{corollary}


\section{Binding Expected Costs}

By Corollary~\ref{c:et:sound}, the expected cost of running $\cmd$
is given by $\ect{\cmd}(\zerof)$.
When $\cmd$ does not contain loops, the latter is easily computable.
To treat loops, \cite{KaminskiKMO16} propose to search for \emph{upper invariants}:
\begin{definition}
  A function $I_f \ofdom \Expect$ is an upper invariant for a loop
  $\whilec[\bexptwo]{\bexp}{\cmd}$ with respect to $f \in \Expect$ if
  \[
    \bracf{\bexptwo \land \bexp} \mulf \et{\cmd}(I_f) \plusf \bracf{\bexptwo \land \neg\bexp} \mulf f \leqf I_f
    \tpkt
  \]
\end{definition}

The following is an application of Park's Theorem with Lemma~\ref{l:et-continuous},
stating that for continuous $F$, $F(I) \leqf I$ implies $\lfp F \leqf I$.
\begin{proposition}[see \cite{KaminskiKMO16}]\label{p:upper-invariant}
  If $I$ is an upper invariant for $\whilec[\bexptwo]{\bexp}{\cmd}$ with respect
  to $f$ then
  \[
    \et[c]{\whilec[\bexptwo]{\bexp}{\cmd}}(f) \leqf I_f
    \tpkt
  \]
\end{proposition}

\section{Modular Runtime via Size Analysis}

Let us denote by $\geq$ the usual product-extension of $\geq$ from $\Realposinf$
to $(\Realposinf)^k$, i.e., $\seq[k]{r} \geq \seq[k]{s}$
if $r_i \geq s_i$ for all $1 \leq i \leq k$.
\begin{definition}
  We call a function $f \ofdom (\Realposinf)^k \to \Realposinf$
  \begin{itemize}
  \item \emph{weakly monotone} if
    \[
      \seq[k]{r} \leq \seq[k]{s} \IImp f(\seq[k]{r}) \leq f(\seq[k]{s})\tkom
    \]
    and
  \item \emph{concave} if
  \[
    p_i \cdot f(\vec{r_i}) \leq f(p_i \cdot \vec{r_i})
  \]
  for all finite or countable infinite families of probabilities $(p_i)_{i \in I}$ ($p_i \geq 0$) with $\sum_{i \in I} p_i \leq 1$
  and vectors $(\vec{r_i})_{i \in I}$ ($\vec{r_i} \in (\Realposinf)^k$).

  \end{itemize}
\end{definition}

For functions $f \ofdom B_1 \times \cdots \times B_k \to C$ and $g_j \ofdom A^l \to B_j$ ($1 \leq j \leq k$)
let
\[
  (f \compose (\seq[k]{g}))(\seq[l]{a}) \defsym f(g_1(\seq[l]{a}),\cdots,g_k(\seq[l]{a}))
  \tpkt
\]

\begin{lemma}\label{l:ect-ect+evt}
  \[
    \ect{\cmd}(f) \leqf \ect{\cmd}(\zerof) + \evt{\cmd}(f)
  \]
\end{lemma}
\begin{proof}
  We prove the stronger claim
  \begin{equation}
    \ect{\cmd}(f \plusf g) \leqf \ect{\cmd}(f) + \evt{\cmd}(g)
    \tpkt
  \end{equation}
\end{proof}

\begin{lemma}\label{l:evt-concave}
  Let $\cmd \in \Cmd$ and let $g \ofdom (\Realposinf)^k \to \Realposinf$
  be a weakly monotone function
  that is concave if $\cmd$ is probabilistic.
  Then
  \[
    \evt{\cmd}(g \compose (\seq[k]{g}))
    \leqf g \compose (\evt{\cmd}(g_1),\dots,\evt{\cmd}(g_k))
    \tpkt
  \]
\end{lemma}

\begin{lemma}\label{l:ect-ect-evt-concave}
  Let $\cmd \in \Cmd$ and let $g \ofdom (\Realposinf)^k \to \Realposinf$
  be a weakly monotone function
  that is concave if $\cmd$ is probabilistic.
  Then
  \[
    \ect{\cmd}(g \compose (\seq[k]{g})) \leqf \ect{\cmd}(\zerof) + g \compose (\evt{\cmd}(g_1),\dots,\evt{\cmd}(g_k))
  \]
\end{lemma}

\begin{theorem}\label{th:comp-seq}
  Let $\cmd,\cmdtwo \in \Cmd$ and let $g \ofdom (\Realposinf)^k \to \Realposinf$
  be a weakly monotone function that is in addition concave if $\cmd$ is probabilistic.

  If $\ect{\cmdtwo}(f) \leqf g \compose (\seq[k]{g})$
  then $\ect{\cmd;\cmdtwo}(f) \leqf \ect{\cmd}(\zerof) \plusf g \compose (\evt{\cmd}(g_1),\dots,\evt{\cmd}(g_k))$.
\end{theorem}

\begin{theorem}\label{th:comp-while}
  Let $\cmd \in \Cmd$ and let $g \ofdom (\Realposinf)^k \to \Realposinf$
  be a weakly monotone function that is in addition concave if $\cmd$ is probabilistic.
  If
  \begin{multline*}
    \bracf{\bexptwo \land \bexp} \mulf \bigl(\ect{\cmd}(\zerof) \plusf g \compose (\evt{\cmd}(g_1),\dots,\evt{\cmd}(g_k))\bigr) \\
    \plusf \bracf{\bexptwo \land \neg\bexp} \mulf f
    \leqf g \compose (\seq[k]{g})
  \end{multline*}
  then $\ect{\whilec[\bexptwo]{\bexp}{\cmd}}(f) \leq g \compose (\seq[k]{g})$.
\end{theorem}
\begin{proof}
  By Lemma~\ref{l:ect-ect-evt-concave}, $g \compose (\seq[k]{g})$ is an upper invariant
  for $\whilec[\bexptwo]{\bexp}{\cmd}$ with respect to $f$.
  Thus, the theorem follows by Proposition~\ref{p:upper-invariant}.
\end{proof}

\section{Implementation}%
\label{s:implementation}

\subsection{Cost Expressions and Constraint Systems}

While keeping program expressions and cost functions abstract for the theoretical development we fix the scope for the implementation.

We define $\Exp$ as terms over integer-valued variables $x \in \Var$, integers $i \in \Z$ and arithmetic functions $\{+,*\}$ representing addition and multiplication.
We define $\BExp$ as inequalities of expressions $a,b \in \Exp$ together with logical connectives $\{\neg, \wedge, \vee\}$ representing logical negation, logical and, and logical or.
\begin{align*}
  a,b             & \bnfdef x             \mid i          \mid a + b                 \mid a * b \\
  \bexp, \bexptwo & \bnfdef a \geqslant b \mid \neg \bexp \mid \bexp \wedge \bexptwo \mid \bexp \vee \bexptwo
\end{align*}
In the implementation we restrict $\DExp$ to be finite distributions over integer expressions $a$,
in notation $\{p_1 \colon a_1,\ldots,p_n \colon a_n\}$.

To provide an intuitive notion of bounds and facilitate automation we introduce \emph{cost expressions}.
\begin{align*}
  m,n             & \bnfdef \cnorm{a} \mid m * n \mid 1 \\
  c,d             & \bnfdef q \cdot m \mid c \cadd d \mid \cmax(c,d) \mid \bracf{\bexp} \cmul c
\end{align*}
We usually write $k$ instead of $k \cmul 1$ for constant expressions $k \in \N$ and $m$ instead of $q \cdot m$ if $q = 1$.
The evaluation function of cost expressions is also denoted by $\eva{\cdot} \colon \CExp \to \Sigma \to \Qpos$.
Notice that $\eva{c} \in \Expect$.

To automate the cost inference of programs we provide a variation of the expectation transformer (cf.~Figure~\ref{fig:et})
$\etcost[c]{\cdot}\ofdom {\Cmd \to \CExp \to \CExp}$
(as well as $\mathsf{ect}^\sharp$ and $\mathsf{evt}^\sharp$).
\begin{figure}
  \centering
  \begin{framed}
    \vspace{-\baselineskip}
    \begin{align*}
      \etcost{\skipc}(f)                               & \defsym f \\
      \etcost[c]{\tickc{r}}(f)                         & \defsym \bracf{c} \cmul r \cadd f \\
      \etcost{\haltc}(f)                               & \defsym 0 \\
      \etcost{\var \passign \{ p_1 \colon a_1, \ldots, p_n \colon a_n\}}(f) & \defsym \textstyle \sum_i^n p_i \cmul f[x/ai] \\
      \etcost{\ifc[\bexptwo]{\bexp}{\cmd}{\cmdtwo}}(f) & \defsym \bracf{\bexptwo \land \bexp} \cmul \etcost{\cmd}(f) \cadd \bracf{\bexptwo \land \neg\bexp} \cmul \etcost{\cmdtwo}(f) \\
      \etcost{\ndc{\cmd}{\cmdtwo}}(f)                  & \defsym \cmax(\etcost{\cmd}(f),\etcost{\cmdtwo}(f)) \\
      \etcost{\pc{p}{\cmd}{\cmdtwo}}(f)                & \defsym p \cmul \etcost{\cmd}(f) \cadd (1-p) \cmul \etcost{\cmdtwo}(f) \\
                                                       & \\
      \etcost{\cmd; \cmdtwo}(f) & \defsym \etcost{\cmd}(\etcost{\cmdtwo}(f)) \\
      \ectcost{\cmd; \cmdtwo}(f) & \defsym \ectcost{\cmd}(0) + \ctxone[\seq{h}]\text{, where}\\
      \shortintertext{%
        \vspace{-\baselineskip}
        \begin{align*}
          \ectcost{\cmdtwo}(\eva{f}) & = C[\seq{g}]                                              \\
          g                          & \defsym \lambda \seq{f} . \eva{C[\seq[k]{f}]} \emph{concave} \\
          \evt{\cmd}(\eva{g_i})      & \leqf \eva{h_i}
      \end{align*}
      }
      \ectcost{\whilec[\bexptwo]{\bexp}{\cmd}}(f) & \defsym \ctxone[\seq[k]{g}] \\
      \shortintertext{%
        \vspace{-\baselineskip}
        \begin{align*}
          g                          & \defsym \lambda \seq{f} . \eva{C[\seq[k]{f}]} \emph{concave} \\
        \evt{\cmd}(\eva{g_i}) & \leqf \eva{h_i} \\
        \bracf{\bexptwo \land \bexp} & \entails \eva{\ectcost{\cmd}(0) + \ctxone[\seq[k]{h}]}\leqf \eva{\ctxone[\seq[k]{g}]} \\
        \wedge \bracf{\bexptwo \land \neg \bexp} & \entails \eva{f} \leqf \eva{\ctxone[\seq[k]{g}]}
      \end{align*}
      }
    \end{align*}
  \end{framed}
  \caption{Definition of expected cost expression transformer $\etcost{\cdot}$.}
  \label{fig:et-cost}
\end{figure}

\begin{theorem}\label{l:ct-gt-et}
  For all commands $\cmd \in \Cmd$ and cost expressions $f \in \CExp$

  $\et{C}(\eva{f}) \leqf \eva{\etcost[c]{C}(f)}$ \tpkt
\end{theorem}
\begin{proof}
  The proof is by induction on the structure of $\cmd$.
  \begin{proofcases}
    \case{$\etcost{\skipc}(f) \defsym f$}
    Then $\et{\skipc}(\eva{f}) = \eva{f} = \eva{\etcost{\skipc}(f)}$.

    \case{$\ectcost{\tickc{r}}(f) \defsym r + f $}
    Then $\ect{\tickc{r}}(\eva{f}) = \constf{r} \plusf \eva{f} = \eva{r} \plusf \eva{f} = \eva{r \cadd f} = \eva{\ectcost{\tickc{r}}(f)}$.

    \case{$\evtcost{\tickc{r}}(f) \defsym f $}
    Then $\evt{\tickc{r}}(\eva{f}) = \eva{f} = \eva{\ectcost{\tickc{r}}(f)}$.

    \case{$\ectcost{\haltc}(f) \defsym 0 $}
    Then $\et{\haltc}(\eva{f}) = \zerof = \eva{0} = \eva{\ectcost{\haltc}(f)}$.

    \case{$\etcost{\var \passign \{p_1 \colon a_1, \ldots, p_n \colon a_n \}}(f) \defsym \sum_i^n p_i \cdot f[x/a_i] $}
    Then%
    \begin{align*}
    \et{\var \passign \{p_1 \colon a_1, \ldots, p_n \colon a_n\}}(\eva{f}) 
      &= \lambda \st. \E{\{ p_1 \colon \eva{a_1}\st, \ldots, p_n \colon \eva{a_n}\st \}}(\lambda i. \eva{f}(\upd{\st}{\var}{i})) \\
      &= \textstyle \lambda \st. \sum_i^n p_i \cdot (\lambda i . \eva{f}(\upd{\st}{\var}{i})(\eva{a_i}\sigma)) \\
      &= \textstyle \sum p_i \cdot \eva{f[x/a_i]} \\
      &= \textstyle \eva{\sum_i^n p_i \cdot f[x/a_i]} \\
      &= \eva{\etcost{\var \passign \{p_1 \colon a_1, \ldots, p_n \colon a_n \}}(f)} \tpkt
    \end{align*}

    \case{$\etcost{\ifc[\bexptwo]{\bexp}{\cmd}{\cmdtwo}}(f) \defsym \bracf{\bexptwo \land \bexp} \mulf \etcost{\cmd}(f) \plusf \bracf{\bexptwo \land \neg\bexp} \mulf \etcost{\cmdtwo}(f)$}
      Then
      \begin{align*}
        \et{\ifc[\bexptwo]{\bexp}{\cmd}{\cmdtwo}}(\eva{f}) 
          & = \bracf{\bexptwo \land \bexp} \mulf \et{\cmd}(\eva{f}) \plusf \bracf{\bexptwo \land \neg\bexp} \mulf \etcost{\cmdtwo}(\eva{f}) \\
          & \leqf \bracf{\bexptwo \land \bexp} \mulf \eva{\etcost{\cmd}(f)} \plusf \bracf{\bexptwo \land \neg\bexp} \mulf \eva{\etcost{\cmdtwo}(f)} \\
          & = \eva{\bracf{\bexptwo \land \bexp} \cmul \etcost{\cmd}(f) \cadd \bracf{\bexptwo \land \neg\bexp} \cadd \etcost{\cmdtwo}(f)} \\
          & = \eva{\etcost{\ifc[\bexptwo]{\bexp}{\cmd}{\cmdtwo}}(f)} \tpkt
      \end{align*}

      \case{$\etcost{\ndc{\cmd}{\cmdtwo}}(f) \defsym \cmax(\etcost{\cmd}(f),\etcost{\cmdtwo}(f))$}
      Then
      \begin{align*}
        \et{\ndc{\cmd}{\cmdtwo}}(\eva{f})
          & = \maxf(\et{\cmd}(\eva{f}),\et{\cmdtwo}(\eva{f}))  \\
          & \leqf \maxf(\eva{\etcost{\cmd}(f)},\eva{\etcost{\cmdtwo}(f)}) \\
          & = \eva{\cmax(\etcost{\cmd}(f),\etcost{\cmdtwo}(f))}
          = \eva{\etcost{\ndc{\cmd}{\cmdtwo}}(f)}
      \end{align*}

      \case{$\etcost{\pc{p}{\cmd}{\cmdtwo}}(f) = p \cmul \etcost{\cmd}(f) \cadd (1-p) \cmul \etcost{\cmd}(f)$}
      Then
      \begin{align*}
        \et{\pc{p}{\cmd}{\cmdtwo}}(\eva{f})
          &= \constf{p} \mulf \et{\cmd}(\eva{f}) \plusf (\constf{1-p}) \mulf \et{\cmdtwo}(\eva{f}) \\
          & \leqf \constf{p} \mulf \eva{\etcost{\cmd}(f)} \plusf (\constf{1-p}) \mulf \eva{\etcost{\cmdtwo}(f)} \\
          &= \eva{ p \cmul \etcost{\cmd}(f) \cadd (1-p) \cmul \etcost{\cmdtwo}(f) }
          = \eva{\etcost{\pc{p}{\cmd}{\cmdtwo}}(f)}
      \end{align*}

      \case{$\etcost{\cmd; \cmdtwo}(f) \defsym \etcost{\cmd}(\etcost{\cmdtwo}(f))$}
      Then
      \begin{align*}
        \et{\cmd; \cmdtwo}(\eva{f})
          &= \et{\cmd}(\et{\cmdtwo}(\eva{f})) \\
          &\leqf \et{\cmd}(\eva{\etcost{\cmdtwo}(f)}) \\
          &\leqf \eva{\etcost{\cmd}(\etcost{\cmdtwo}(f))}
          = \eva{\etcost{\cmd; \cmdtwo}(f)}
      \end{align*}

      \case{$\ectcost{\cmd; \cmdtwo}(f) \defsym \ectcost{\cmd}(0) + \ctxone[\seq{h}] $}
      Then.

      By IH 
      $\ect{\cmdtwo}(\eva{f}) \leqf \eva{\ectcost{\cmdtwo}(f)} 
      = \eva{\ctxone[\seq[k]{g}]} 
      = g \compose (\eva{g_1},\ldots,\eva{g_k})$
      
      By Thorem~\ref{th:comp-seq}.
      \begin{align*}
        \ect{\cmd; \cmdtwo}(\eva{f})
        &= \ect{\cmd}(\zerof) \plusf g \compose (\evt{\cmd}(\eva{g_1}),\ldots,\evt{\cmd}(\eva{g_k})) \\
        &\leqf \eva{\ectcost{\cmd}(0)} \plusf g \compose (\eva{h_1},\ldots,\eva{h_k}) \\
        &= \eva{\ectcost{\cmd}(0)} \plusf \eva{\ctxone[h_1,\ldots,h_k]} \\
        &= \eva{\ectcost{\cmd}(0) \cadd \ctxone[h_1,\ldots,h_k]}
        = \eva{\ectcost{\cmd; \cmdtwo}(f)}
      \end{align*}

      \case{$\ectcost{\whilec[\bexptwo]{\bexp}{\cmd}}(f) \defsym \ctxone[\seq[k]{g}]$}

      By assumption.
      \begin{alignat*}{2}
        & \bracf{\bexptwo \land \bexp} && \entails \eva{\ectcost{\cmd}(0) + \ctxone[\seq[k]{h}]}\leqf \eva{\ctxone[\seq[k]{g}]} \\
        \wedge & \bracf{\bexptwo \land \neg \bexp} && \entails \eva{f} \leqf \eva{\ctxone[\seq[k]{g}]} \\
               &&& \mathpar{By assumption $g \defsym \lambda \seq{f} . \eva{C[\seq[k]{f}]}$ concave.} \\
                 & \bracf{\bexptwo \land \bexp} && \entails \eva{\ectcost{\cmd}(0)} \plusf g \compose (\eva{h_1},\ldots,\eva{h_k}) \leqf g \compose (\eva{g_1},\ldots,\eva{g_k}) \\
               \wedge & \bracf{\bexptwo \land \neg \bexp} && \entails \eva{f} \leqf g \compose (\eva{g_1},\ldots,\eva{g_k}) \\
                      &&& \mathpar{By IH $\ect{\cmd}(\zerof) \leqf \ectcost{\cmd}(0)$. By assumption $\evt{\cmd}(\eva{g_i} \leqf \eva{h_i})$. } \\
                 & \bracf{\bexptwo \land \bexp} && \entails \ect{C}(\zerof) \plusf g \compose (\evt{\cmd}(\eva{g_1}),\ldots,\evt{\cmd}(\eva{g_k})) \leqf g \compose (\eva{g_1},\ldots,\eva{g_k}) \\
               \wedge & \bracf{\bexptwo \land \neg \bexp} && \entails \eva{f} \leqf g \compose (\eva{g_1},\ldots,\eva{g_k})
      \end{alignat*}
      By Theorem~\ref{th:comp-while}.
      \begin{align*}
        \ect{\whilec[\bexptwo]{\bexp}{\cmd}}(\eva{f})
        & \leqf g \compose (\eva{g_1},\ldots,\eva{g_k})  \\
        & = \eva{\ctxone[\seq[k]{g}]} 
        = \eva{\ectcost{\whilec[\bexptwo]{\bexp}{\cmd}}(f)}
      \end{align*}
  \end{proofcases}
\end{proof}

\bibliographystyle{plain}
\bibliography{references}
\end{document}
\grid